\documentclass[journal,onecolumn]{IEEEtran}
\usepackage{cite}
\usepackage{amsmath}
\usepackage{amsfonts,mathrsfs,latexsym,amssymb}
\usepackage{algorithmic}
\usepackage{array}
\usepackage{url}
\usepackage{longtable}
\usepackage{supertabular}
\usepackage{multirow}%
\usepackage{tabularx}%
\usepackage{caption}%
\usepackage{float}%

\newtheorem{theorem}{Theorem}

\newtheorem{lemma}[theorem]{Lemma}
\newtheorem{corollary}[theorem]{Corollary}

\begin{document}
%
\title{Constructions of Snake-in-the-Box Codes under the $\ell_{\infty}$-Metric for Rank Modulation}
%
%

\author{Xiang~Wang
        and~Fang-Wei~Fu
\thanks{X. Wang and F.-W. Fu are with the Chern Institute of Mathematics and LPMC, Nankai
University, Tianjin 300071, China (e-mail: xqwang@mail.nankai.edu.cn; fwfu@nankai.edu.cn).}
}

%



\maketitle

\begin{abstract}
In the rank modulation scheme, Gray codes are very useful in the realization of flash memories. For a Gray code in this scheme, two adjacent codewords are obtained by using one ``push-to-the-top'' operation. Moreover, snake-in-the-box codes under the $\ell_{\infty}$-metric are Gray codes, which can be capable of detecting one $\ell_{\infty}$-error. In this paper, we give two constructions of $\ell_{\infty}$-snakes. On the one hand, inspired by Yehezkeally and Schwartz's construction, we present a new construction of the $\ell_{\infty}$-snake. The length of this $\ell_{\infty}$-snake is longer than  the length of the $\ell_{\infty}$-snake constructed by Yehezkeally and Schwartz. On the other hand, we also give another construction of $\ell_{\infty}$-snakes by using $\mathcal{K}$-snakes and obtain the longer $\ell_{\infty}$-snakes than the previously known ones.
\end{abstract}

\begin{IEEEkeywords}
 Flash memory, rank modulation, Gray codes, snake-in-the-box codes, $\mathcal{K}$-snakes, $\ell_{\infty}$-snakes.
\end{IEEEkeywords}

%
\IEEEpeerreviewmaketitle

\section{Introduction}
\label{sec1}
\IEEEPARstart{F}{lash} memory is a non-volatile storage medium that is both electrically programmable and erasable. It has been widely used because of its reliability, relative low cost, and high storage density. In flash memories, a block which contains many cells can maintain a block of charge levels to represent information. However, the flash memory has its inherent asymmetry between cell programming (injecting cells with charge) and  cell erasure (removing charge from cells). That is to say, increasing the charge level of a single cell (cell programming) is an easy operation, but decreasing the charge level of a single cell (cell erasure) is a very difficult process. In the current flash memories' architecture, a single-cell erasure operation includes copying a large whole block which contains the single cell to a temporary location, erasing it, and then reprogramming all the cells in the block. In the programming operation, some cells may be injected with extra charge. This will lead to overshooting of charge. Then, overprogramming (overshooting of charge) is a severe problem because of some very difficult cell erasure operations. Thus, in order to avoiding overprogramming, injecting a cell with charge by many iterations makes its charge level gradually approach the desirable level in the cell programming operation. Due to charge leakage or reading disturbance, these may result in some errors in flash memory cells.

The rank modulation scheme has been recently proposed in \cite{Jiang1} to overcome these outstanding problems. In this scheme, one permutation is induced by relative rankings of the charge levels on a group of cells instead of using absolute values of charge levels. Moreover, this permutation is used to represent information. Specifically, assume that $c_1,c_2,...,c_n\in\mathbb{R}$ represent the charge levels of $n\in \mathbb{N}$ cells respectively, then the charge levels of these cells induce one permutation $\pi=[\pi(1),...,\pi(n)]\in S_n$ such that $c_{\pi(1)}>c_{\pi(2)}>\cdot\cdot\cdot>c_{\pi(n)}$, where $S_n$ is the set of all the permutations over $\{1,2,...,n\}$. In order to avoid the overprogramming, the cell programming uses only ``push-to-the-top'' operations \cite{Jiang1}. That is, a cell is programmed by raising the charge level of this cell above those of all others in the block. Therefore, in the manner, the overprogramming is no longer a problem. Moveover, when injection of some extra charge or leakage may not change their relative rankings, then the permutation induced by their relative rankings will not change. Hence, this will not cause an encoding error. However, when the relative rankings are changed because of injection of much extra charge or leakage in the cells, the permutation induced by the relative rankings will be different from the desired permutation, i.e., this leads to an encoding error. To detect and/or correct such errors, we require an appropriate distance measure in the permutations. There are several metrics on the permutations such as the $\ell_{\infty}$-metric \cite{Klve}, \cite{Tamo}, the Ulam metric \cite{Farnoud}, and the Kendall's $\tau$-metric \cite{Jiang2}, \cite{Barg}, \cite{Mazumdar}. In this paper, we will consider only the $\ell_{\infty}$-metric.

In the rank modulation scheme, Gray codes are important codes which represent information in flash memories. In \cite{Jiang1}, Jiang \emph{et al.} proposed the Gray codes by using ``push-to-the-top'' operations. For the Gray code, it was first presented in \cite{Gray}, where it is a sequence of distinct binary vectors of fixed length and every adjacent pair differs in a single coordinate. In practice, they are widely used in many applications \cite{Savage1}. Moreover, a good survey on the Gray codes is given in \cite{Savage2}. Recently, Gray codes for rank modulation have been studied in \cite{Jiang2, Yehezkeally, Gad1, Gad2}. In addition, a snake-in-the-box code is a Gray code in which the distance of any two distinct codewords in the code is at least $2$. Thus, this code can detect a single error in one codeword. In particular, the snake-in-the-box codes are usually studied in the context of binary codes in the Hamming scheme, e.g.\cite{Abbott}. In this paper, we will only focus on the snake-in-the-box codes under the $\ell_{\infty}$-metric.

In \cite{Yehezkeally}, Yehezkeally and Schwartz constructed directly a snake-in-the-box code of length $\lceil\frac{n}{2}\rceil!(\lfloor\frac{n}{2}\rfloor+(\lfloor\frac{n}{2}\rfloor-1)!)$ in $S_{n}$ under the $\ell_{\infty}$-metric. In this paper, we will improve on this result. On the one hand, we will construct a snake of length $\lceil\frac{n}{2}\rceil!(\lfloor\frac{n}{2}\rfloor+(\lfloor\frac{n}{2}\rfloor)!)$ in $S_{n}$. On the other hand, we will also construct the longer $\ell_{\infty}$-snakes by using $\mathcal{K}$-snakes.

The rest of this paper is organized as follows. In Section \ref{sec2}, we will give some basic definitions for the rank modulation scheme and notations required in this paper. In Section \ref{sec3}, we give directly two constructions of $\ell_{\infty}$-snakes in $S_n$. In Section \ref{sec4}, we present some examples of these two constructions. In Section \ref{sec5}, we compare our results with the previous ones. Section \ref{sec6} concludes this paper.

\section{Preliminaries}
\label{sec2}
In this section, we will use some definitions and notations mentioned in \cite{Horvitz} and \cite{Buzaglo}.

We let $[n]\triangleq\{1,2,...,n\}$ and let $\pi\triangleq[\pi(1),\pi(2),...,\pi(n)]$ be a \emph{permutation} over $[n]$. And let $S_n$ be the set of all the permutations over $[n]$. For $\sigma,\pi \in S_n$, their multiplication $\pi\circ\sigma$ is denoted by the composition of $\sigma$ on $\pi$, i.e., $\pi\circ\sigma(i)=\sigma(\pi(i))$, for all $i\in [n]$. Under this multiplication operation, $S_n$ is a noncommunicative \emph{group}. Moreover, let $\pi^{-1}$ be the \emph{inverse} element of $\pi$, for $\pi\in S_n$, and let $A_n$ be the subgroup of all \emph{even} permutations over $[n]$.

Assume that given a set $\mathcal{S}$ and a set of transformations $T\subset\{f|f:\mathcal{S}\rightarrow\mathcal{S}\}$, a \emph{Gray code} over $\mathcal{S}$ of size $M$, is a sequence $C=(c_0,c_1,...,c_{M-1})$ of $M$ different elements from $\mathcal{S}$, called \emph{codewords}, in which for each $i\in[M-1]$ there exists some $\tilde{t}_i\in T$ such that $c_i=\tilde{t}_i(c_{i-1})$. For convenience, we denote a \emph{transformation sequence} of the Gray code $C$ by $\mathcal{T}_C$, i.e., $\mathcal{T}_C=(\tilde{t}_1,\tilde{t}_2,...,\tilde{t}_{M-1})$. The Gray code is called \emph{complete} if $M=|\mathcal{S}|$, and \emph{cyclic} if there exists $\tilde{t}_{M}\in T$ such that $c_0=\tilde{t}_{M}(c_{M-1})$. Thus, the transformation sequence $\mathcal{T}_C$ of the cyclic Gray code $C$ is $(\tilde{t}_1,\tilde{t}_2,...,\tilde{t}_{M-1},\tilde{t}_{M})$.

Consider the Gray codes for rank modulation in flash memories, we have $\mathcal{S}=S_n$ and the set of transformations comprises of all the ``push-to-the-top'' operations in $S_n$, defined by $T_n$. Next, we denote by $t_i:S_n\rightarrow S_n$ one ``push-to-the-top'' operation on index $i$, for $2\leq i\leq n$, that is,
\begin{align}
t_i[a_1,a_2,....,a_{i-1},a_i,a_{i+1},...,a_n]=[a_i,&a_1,a_2,...,a_{i-1},a_{i+1},...,a_n],\nonumber
\end{align}
and a p-transition will be an abbreviation of a ``push-to-the-top'' operation. Therefore, $T_n=\{t_2,t_3,...,t_n\}$.

A sequence of p-transitions is called a \emph{transition sequence}. Given an initial permutation $\pi_0$ in $S_n$ and a transition sequence $(t_{x(1)},t_{x(2)},...,t_{x(l)})$ with $x(i)\in[n]$ for all $i\in[l]$, we can obtain a sequence of permutations $\pi_0,\pi_1,...,\pi_l$ in $S_n$, where $\pi_i=t_{x(i)}(\pi_{i-1})$ for all $i\in[l]$. When $\pi_l=\pi_0$ and $\pi_i\neq\pi_j$ for each pair $0\leq i<j<l$, the permutation sequence $(\pi_0,\pi_1,...,\pi_{l-1})$ is a cyclic Gray code by using the ``push-to-the-top'' operations, denoted by $C_n$. Moreover, the transition sequence $\mathcal{T}_{C_n}$ is $(t_{x(1)},t_{x(2)},...,t_{x(l)})$. For convenience, for a transition sequence $\mathcal{T}$ in $T_n$, we denote by $f:\{1,2,...,|\mathcal{T}|\}\rightarrow\{2,3,...,n\}$ one index function of the transition sequence. Hence, $\mathcal{T}=(t_{f(1)},t_{f(2)},...,t_{f(|\mathcal{T}|)})$.

The $\ell_{\infty}$-distance between two permutations $\pi,\sigma\in S_n$, denoted by $d_{\infty}(\pi,\sigma)$, is the maximal number of indices difference between $\pi$ and $\sigma$. A snake-in-the-box code $C$ under the $\ell_{\infty}$-metric is a Gray code in which for each two distinct permutations $\pi,\sigma\in C$, we have that $d_{\infty}(\pi,\sigma)\geq2$. Moreover, a cyclic $\ell_{\infty}$-snake is a cyclic Gray code. If not special specified, we call a cyclic $\ell_{\infty}$-snake an $\ell_{\infty}$-snake or a snake.

Furthermore, we denote by an $(n,M,\ell_{\infty})$-snake an $\ell_{\infty}$-snake of size $M$ in $S_n$. And we let $C_{\mathcal{T}_{C}}^{\pi_0}$ be an $(n,M,\ell_{\infty})$-snake, where $\mathcal{T}_C$ is its transition sequence and $\pi_0$ is its first permutation. For simplicity, we let $C_{\mathcal{T}_{C}}^{\pi_0}\triangleq(\pi_0,\pi_1,\pi_2,...,\pi_{M-1})$ and $\mathcal{T}_{C}\triangleq(t_{x(1)},t_{x(2)},...,t_{x(M)})$ such that $\pi_i=t_{x(i)}(\pi_{i-1})$ for every $i\in [M-1]$ and $t_{x(M)}(\pi_{M-1})=\pi_0$.

In \cite{Jiang1}, Jiang \emph{et al.} presented an $n$-length rank modulation Gray code ($n$-RMGC) for flash memories in the rank modulation scheme. For convenience, we denote a cyclic and complete $n$-RMGC by $C_{\mathcal{T}_{n}}$, where $\mathcal{T}_{n}$ is its transition sequence. Hence, $|C_{\mathcal{T}_{n}}|=n!$. And we define $\mathcal{T}_{n}\triangleq(t_{i_{n}(1)},t_{i_{n}(2)}...,t_{i_{n}(n!)})$. In \cite{Tamo}, Tamo and Schwartz discussed error-correcting and error-detecting codes in $S_n$ under the $\ell_{\infty}$-metric, called as the limited-magnitude rank-modulation codes (LMRM codes). Moreover, in \cite{Tamo}, they also proved that if $C$ is an $(n,M,\ell_{\infty})$-snake, then $M\leq \frac{n!}{2^{\lfloor n/2\rfloor}}$. Later Yehezkeally and Schwartz \cite{Yehezkeally} constructed an $(n,M,\ell_{\infty})$-snake of size $\lceil\frac{n}{2}\rceil!(\lfloor\frac{n}{2}\rfloor+(\lfloor\frac{n}{2}\rfloor-1)!)$ for all $n\geq 4$.

Having the above definitions and notations, we will present two constructions of $\ell_{\infty}$-snakes in the following section.

\section{Main results}
\label{sec3}
According to the definition of the $\ell_{\infty}$-distance, for any two permutations $\sigma,\pi\in S_n$, we have the following expression for $d_{\infty}(\sigma,\pi)$ \cite{Yehezkeally},
\begin{equation}
d_{\infty}(\sigma,\pi)=\max\limits_{i\in[n]}|\sigma(i)-\pi(i)|\nonumber.
\end{equation}

\subsection{Construction of $\ell_{\infty}$-snakes by using cyclic and complete RMGCs}
\label{subsec1}
In this subsection, we give one construction of $\ell_{\infty}$-snakes by using cyclic and complete RMGCs. In order to use the code constructions presented in \cite{Jiang1}, we will give the following lemma.

\begin{lemma}\cite[Theorems 7 and 8]{Jiang1}
\label{Lm 1}
For all $n\geq 3$, there exists a cyclic and complete $(n-1)$-RMGC, denoted by $C_{\mathcal{T}_{n-1}}$, where $\mathcal{T}_{n-1}\triangleq(t_{i_{n-1}(1)},t_{i_{n-1}(2)},...,t_{i_{n-1}((n-1)!)})$. Moreover, for all $n\geq 4$, the constructions in \cite[Theorem 7]{Jiang1} can yield a cyclic and complete $n$-RMGC, denoted by $C_{\mathcal{T}_{n}}$, with its transition sequence $\mathcal{T}_{n}=(t_{i_{n}(1)},t_{i_{n}(2)}...,t_{i_{n}(n!)})$, where
\begin{align*}
\mathcal{T}_n=(\underbrace{t_n,...,t_n}_{n-1},t_{\hat{i}_{n-1}(1)},\underbrace{t_n,...,t_n}_{n-1},&t_{\hat{i}_{n-1}(2)},...,
\underbrace{t_n,...,t_n}_{n-1},t_{\hat{i}_{n-1}((n-1)!)})
\end{align*}
and
\begin{equation}
t_{\hat{i}_{n-1}(j)}=t_{n-i_{n-1}(j)+1}\quad \text{for all $j\in[(n-1)!]$}.\nonumber
\end{equation}
\end{lemma}

According to the above lemma, we can obtain some properties of this RMGC which we will later use.

\begin{lemma}
\label{Lm 2}
For any $n\geq 3$, there exists a cyclic and complete $n$-RMGC, denoted by $C_{\mathcal{T}_{n}}$, where its transition sequence $\mathcal{T}_{n}=(t_{i_{n}(1)},...,t_{i_{n}(n!)})$ such that
\begin{equation}
t_{i_{n}(j)}=t_2,~t_{i_{n}(k)}=t_{n-1},~\text{and}~ t_{i_{n}(l)}=t_n \nonumber
\end{equation}
for some $j,k,l\in[n!]$.
\end{lemma}
\begin{IEEEproof}
According to the construction of \cite[Fig. 2]{Jiang1}, we have one transition sequence of a cyclic and complete $3$-RMGC, denoted by $\mathcal{T}_{3}$, where $\mathcal{T}_{3}=(t_3,t_3,t_2,t_3,t_3,t_2)$. Hence, we have that
\begin{equation}
t_{{i}_{3}(1)}=t_3~\text{and}~t_{{i}_{3}(3)}=t_{{i}_{3}(6)}=t_2.\label{Eq113}
\end{equation}
By Lemma \ref{Lm 1}, when $n=4$, we have a cyclic and complete $4$-RMGC, denoted by $C_{\mathcal{T}_{4}}$, with
its transition sequence $\mathcal{T}_{4}=(t_{i_4(1)},...,t_{i_{4}(4!)})$, where
\begin{equation}
\mathcal{T}_{4}=(\underbrace{t_4,...,t_4}_{3},t_{{\hat{i}}_{3}(1)},...,
\underbrace{t_4,...,t_4}_{3},t_{{\hat{i}}_{3}(3!)}).\label{Eq 1}
\end{equation}
By Lemma \ref{Lm 1} and $(\ref{Eq113})$, we have that
\begin{equation}
t_{{\hat{i}}_{3}(1)}=t_2~\text{and}~t_{{\hat{i}}_{3}(3)}=t_3.\label{Eq 2}
\end{equation}
Hence, by $(\ref{Eq 1})$ and $(\ref{Eq 2})$, we can obtain that
\begin{equation}
t_{{i}_{4}(4)}=t_2,~t_{{i}_{4}(12)}=t_3,~\text{and}~t_{i_{4}(1)}=t_4.\label{Eq 3}
\end{equation}
By Lemma \ref{Lm 1}, we can also get a cyclic and complete $5$-RMGC, denoted by $C_{\mathcal{T}_{5}}$, with its transition sequence $\mathcal{T}_{5}=(t_{i_{5}(1)},...,t_{i_{5}(5!)})$, where
\begin{equation}
\mathcal{T}_{5}=(\underbrace{t_5,...,t_5}_{4},t_{{\hat{i}}_{4}(1)},...,\underbrace{t_5,...,t_5}_{4},t_{{\hat{i}}_{4}(4!)})\nonumber
\end{equation}
and
\begin{equation}
t_{\hat{i}_{4}(j)}=t_{5-i_{4}(j)+1}\quad \text{for all $j\in[4!]$}.\label{Eq 4}
\end{equation}
By $(\ref{Eq 3})$ and $(\ref{Eq 4})$, we have that
\begin{equation}
t_{i_{5}(20)}=t_{\hat{i}_{4}(4)}=t_{5-i_{4}(4)+1}=t_4\nonumber
\end{equation}
and
\begin{equation}
t_{i_{5}(5)}=t_{\hat{i}_{4}(1)}=t_{5-i_{4}(1)+1}=t_2\nonumber.
\end{equation}
Hence, we have that
\begin{equation}
t_{{i}_{5}(5)}=t_2,~t_{{i}_{5}(20)}=t_4,~\text{and}~t_{i_{5}(1)}=t_5.\nonumber
\end{equation}

Similarly, by induction, we can obtain that there exists a cyclic and complete $n$-RMGC, denoted by $C_{\mathcal{T}_{n}}$, with its transition sequence $\mathcal{T}_{n}=(t_{i_{n}(1)},...,t_{i_{n}(n!)})$ such that
\begin{equation}
t_{{i}_{n}(n)}=t_2,~t_{{i}_{n}(n^2-n)}=t_{n-1},~\text{and}~t_{i_{n}(1)}=t_n\nonumber
\end{equation}
for all $n\geq 3$.
\end{IEEEproof}

The following lemma gives one construction of a basic block which is useful for the construction of $\ell_{\infty}$-snakes by using cyclic and complete RMGCs.

\begin{lemma}
\label{Lm 3}
For all $n\geq 6$, let $\{a_{j}\}_{j=1}^{k}$ be a set of even integers of $[n]$ and $\{b_{j}\}_{j=1}^{l}$ be a set of odd integers of $[n]$, where $k=\lfloor\frac{n}{2}\rfloor$ and $l=\lceil\frac{n}{2}\rceil$. And let $\sigma=[b_1,a_2,a_3...,a_k,a_1,b_2,b_3,...,b_l]$ be a permutation such that $|a_1-b_1|\geq 2$. Then, there exist two noncyclic $(n,k!+k,\ell_{\infty})$-snakes. One noncyclic $(n,k!+k,\ell_{\infty})$-snake, denoted by $C_{\mathcal{T}_C}^{\sigma,\pi_1}$, is starting with $\sigma$ and ending with one permutation $\pi_1$, where
\begin{equation}
\pi_1=[a_2,a_3,...,a_{k-1},a_k,a_1,b_1,b_2,...,b_l].\nonumber
\end{equation}
Another noncyclic $(n,k!+k,\ell_{\infty})$-snake, denoted by $\hat{C}_{\mathcal{T}_{\hat{C}}}^{\sigma,\pi_2}$, is starting with $\sigma$ and ending with one permutation $\pi_2$, where
\begin{equation}
\pi_2=[a_2,a_3,...,a_{k-1},a_1,a_k,b_1,b_2,...,b_l].\nonumber
\end{equation}
\end{lemma}

\begin{proof}
We prove only the existence of $C_{\mathcal{T}_C}^{\sigma,\pi_1}$, since the proof of the existence of $\hat{C}_{\mathcal{T}_{\hat{C}}}^{\sigma,\pi_2}$ is similar. For convenience, let $C_{\mathcal{T}_C}^{\sigma,\pi_1}\triangleq(\sigma_0,\sigma_1,...,\sigma_{k!+k-1})$ and  $\mathcal{T}_C\triangleq(t_{{\alpha}_{1}(1)},t_{{\alpha}_{1}(2)},...,t_{{\alpha}_{1}(k!+k-1)})$.

Now, by Lemma \ref{Lm 1}, there exists a cyclic and complete $k$-RMGC with its transition sequence $\mathcal{T}_k$, where
\begin{equation}
\mathcal{T}_k=(t_{i_{k}(1)},t_{i_{k}(2)},...,t_{i_{k}(k!)}).\label{Eq 5}
\end{equation}
By Lemma \ref{Lm 2}, since $k\geq 3$, we have that
\begin{equation}
t_{i_{k}(s_1)}=t_{k}~ \text{and}~t_{i_{k}(s_2)}=t_{k-1}~~\text{for some $s_1,s_2\in [k!]$}.\label{Eq 6}
\end{equation}
Hence, by $(\ref{Eq 5})$ and $(\ref{Eq 6})$, we can obtain two transition sequences, denoted by $\mathcal{T}_{k}^{1}$ and $\mathcal{T}_{k}^{2}$, where
\begin{equation}
\mathcal{T}_{k}^{1}=(t_{i_{k}(s_1+1)},t_{i_{k}(s_1+2)},...,t_{i_{k}(k!)},t_{i_{k}(1)},t_{i_{k}(2)},...,
t_{i_{k}(s_1)})\nonumber
\end{equation}
and
\begin{equation}
\mathcal{T}_{k}^{2}=(t_{i_{k}(s_2+1)},t_{i_{k}(s_2+2)},...,t_{i_{k}(k!)},t_{i_{k}(1)},t_{i_{k}(2)},...,
t_{i_{k}(s_2)}).\nonumber
\end{equation}
For convenience, we let $\mathcal{T}_{k}^{j}\triangleq(t_{\beta_{j}(1)},t_{\beta_{j}(2)},...,t_{\beta_{j}(k!)})$ for $j=1,2$. Applying some transition sequence $\mathcal{T}_{k}^{j}$ on one initial permutation $\hat{\pi}$, where $\hat{\pi}\in S_k$ and $\hat{\pi}=[c_1,c_2,...,c_k]$, then we can obtain a cyclic and complete $k$-RMGC, denoted by $C_{\mathcal{T}_{k}^{j}}^{\hat{\pi}}$, with its last permutation $\tilde{\pi}_j$ for $j=1,2$. By the construction of $\mathcal{T}_{k}^{j}$, when $j=1$, we have that
\begin{equation}
\tilde{\pi}_1=[c_2,c_3,...,c_{k-1},c_k,c_1].\label{Eq 7}
\end{equation}
And when $j=2$, we have that
\begin{equation}
\tilde{\pi}_2=[c_2,c_3,...,c_{k-1},c_1,c_k].\nonumber
\end{equation}

Next, we construct the transition sequence of $C_{\mathcal{T}_C}^{\sigma,\pi_1}$. We let $\sigma_0\triangleq\sigma$, then $\sigma_0=[b_1,a_2,...,a_k,a_1,b_2,...,b_l]$. When $1\leq j\leq k-1$, we let $t_{\alpha_{1}(j)}=t_k$. When $j=k$, we let $t_{\alpha_{1}(k)}=t_{k+1}$. If $k+1\leq j\leq k!+k-1$, we use the transition sequence $\mathcal{T}_{k}^{1}$ to construct the p-transition $t_{\alpha_{1}(j)}$, and let $t_{\alpha_{1}(j)}=t_{\beta_{1}(j-k)}$. Hence, we have that $\sigma_j=t_{\alpha_{1}(j)}(\sigma_{j-1})$ for all $1\leq j\leq k!+k-1$.

Finally, we will prove that for any $0\leq i<j\leq k!+k-1$, we have that $d_{\infty}(\sigma_i,\sigma_j)\geq 2$. By the construction of $t_{\alpha_{1}(j)}$, when $1\leq j\leq k-2$, we have that
\begin{equation}
\sigma_j=[a_{k+1-j},..,a_k,b_1,a_2,...,a_{k-j},a_1,b_2,...,b_l].\nonumber
\end{equation}
When $j=k-1$, we have that
\begin{equation}
\sigma_{k-1}=[a_2,..,a_k,b_1,a_1,b_2,...,b_l].\nonumber
\end{equation}
When $j=k$, we have that
\begin{equation}
\sigma_k=[a_1,a_2,...,a_k,b_1,b_2,...,b_l].\label{Eq 8}
\end{equation}
By $(\ref{Eq 7})$ and $(\ref{Eq 8})$, we can obtain that
\begin{equation}
\pi_1=\sigma_{k!+k-1}=[a_2,...,a_k,a_1,b_1,...,b_l]. \label{Eq 9}
\end{equation}
When $0\leq i<j\leq k-1$, we obtain easily that
\begin{equation}
d_{\infty}(\sigma_i,\sigma_j)\geq 2. \label{Eq 10}
\end{equation}
When $0\leq i\leq k-1~\text{and}~k\leq j\leq k!+k-1$, we have $\sigma_i(k+1)=a_1~\text{and}~\sigma_j(k+1)=b_1$, then
\begin{align}
d_{\infty}(\sigma_i,\sigma_j)\geq&|\sigma_i(k+1)-\sigma_j(k+1)|\nonumber\\
=&|a_1-b_1|\nonumber\\
\geq&2.\label{Eq 11}
\end{align}
When $k\leq i<j\leq k!+k-1$, we know that the first $k$ elements of $\sigma_i$ and $\sigma_j$ are different permutations over $\{a_j\}_{j=1}^{k}$. Since $\{a_j\}_{j=1}^{k}$ is a set of even integers, then
\begin{equation}
d_{\infty}(\sigma_i,\sigma_j)\geq 2.\label{Eq 12}
\end{equation}
Hence, by $(\ref{Eq 9})-(\ref{Eq 12})$, we can obtain a noncyclic $(n,k!+k,\ell_{\infty})$-snake $C_{\mathcal{T}_C}^{\sigma,\pi_1}$ starting with $\sigma$ and ending with $\pi_1=[a_2,a_3,...,a_k,a_1,b_1,b_2,...,b_l]$.

Similarly, we can construct another noncyclic $(n,k!+k,\ell_{\infty})$-snake $\hat{C}_{\mathcal{T}_{\hat{C}}}^{\sigma,\pi_2}$. Let $\mathcal{T}_{\hat{C}}\triangleq(t_{{\alpha}_{2}(1)},t_{{\alpha}_{2}(2)},...,t_{{\alpha}_{2}(k!+k-1)})$ and $\hat{C}_{\mathcal{T}_{\hat{C}}}^{\sigma,\pi_2}\triangleq(\hat{\sigma}_0,\hat{\sigma}_1,...,\hat{\sigma}_{k!+k-1})$. Analogously, when $1\leq j\leq k-1$, we let $t_{\alpha_{2}(j)}=t_k$. When $j=k$, we let $t_{\alpha_{2}(k)}=t_{k+1}$. If $k+1\leq j\leq k!+k-1$, we use the transition sequence $\mathcal{T}_{k}^{2}$ to construct the transition $t_{\alpha_{2}(j)}$, and let $t_{\alpha_{2}(j)}=t_{\beta_{2}(j-k)}$. Moreover, we let ${\hat{\sigma}}_0=\sigma$. Then, we have that $\hat{\sigma}_j=t_{\alpha_{2}(j)}(\hat{\sigma}_{j-1})$ for all $1\leq j\leq k!+k-1$. As the above discussion, we can also obtain another noncyclic $(n,k!+k,\ell_{\infty})$-snake $\hat{C}_{\mathcal{T}_{\hat{C}}}^{\sigma,\pi_2}$ starting with $\sigma$ and ending with $\pi_2=[a_2,a_3,...,a_{k-1},a_1,a_k,b_1,b_2,...,b_l]$.
\end{proof}

In the following, by Lemma \ref{Lm 3}, we will give the construction of an $(n,M,\ell_{\infty})$-snake of size
$M=\lceil\frac{n}{2}\rceil!(\lfloor\frac{n}{2}\rfloor+\lfloor\frac{n}{2}\rfloor!)$. Suppose $p\triangleq\lceil\frac{n}{2}\rceil$ and $q\triangleq\lfloor\frac{n}{2}\rfloor$, then $[n]$ has $p$ odd elements and $q$ even ones. We let $E_n$ be the set of even integers of $[n]$, and let $O_n$ be the set of odd integers of $[n]$. Consider $n\geq 6$, we let $\sigma_0$ be the first permutation of the $\ell_{\infty}$-snake, where
\begin{equation}
\sigma_0=[1,4,...,2q-2,2,2q,3,5...,2p-1].\label{Eq 13}
\end{equation}

Firstly, we construct one transition sequence, denoted by $\mathcal{T}=\{t_{i(1)},t_{i(2)},...,t_{i(M)}\}$. By Lemma \ref{Lm 1}, we take a cyclic and complete $p$-RMGC by using the following transition sequence
\begin{equation}
\mathcal{T}_p=(t_{i_{p}(1)},t_{i_{p}(2)},...,t_{i_{p}(p!)}).\label{Eq 14}
\end{equation}
Moreover, by Lemma \ref{Lm 3}, we can obtain two noncyclic $(n,M_q,\ell_{\infty})$-snakes of size $M_q=q!+q$, denoted by
$C_{\mathcal{T}_C}^{\sigma,\pi_1}$ and $\hat{C}_{\mathcal{T}_{\hat{C}}}^{\sigma,\pi_2}$, respectively. And $C_{\mathcal{T}_C}^{\sigma,\pi_1}$ is given by the following transition sequence
\begin{equation}
\mathcal{T}_C=(t_{\alpha_{1}(1)},t_{\alpha_{1}(2)},...,t_{\alpha_{1}(M_q-1)}) \label{Eq 15}
\end{equation}
starting with $\sigma=[b_1,a_2,...,a_{q-1},a_q,a_1,b_2,...,b_p]$ and ending with $\pi_1=[a_2,...,a_{q-1},a_q,a_1,b_1,b_2,...,b_p]$, where $\sigma,\pi_1\in S_n,~E_n=\{a_j|j=1,...,q\},~ O_n=\{b_j|j=1,...,p\},~\text{and}~|a_1-b_1|\geq 2$. Similarly, $\hat{C}_{\mathcal{T}_{\hat{C}}}^{\sigma,\pi_2}$ is determined by the following transition sequence
\begin{equation}
\mathcal{T}_{\hat{C}}=(t_{\alpha_{2}(1)},t_{\alpha_{2}(2)},...,t_{\alpha_{2}(M_q-1)}) \label{Eq 16}
\end{equation}
starting with $\sigma=[b_1,a_2,...,a_{q-1},a_q,a_1,b_2,...,b_p]$ and ending with $\pi_2=[a_2,...,a_{q-1},a_1,a_q,b_1,b_2,...,b_p]$, where $\sigma,\pi_2\in S_n,~E_n=\{a_j|j=1,...,q\},~ O_n=\{b_j|j=1,...,p\},~\text{and}~|a_1-b_1|\geq 2$.

By $(\ref{Eq 13})-(\ref{Eq 16})$, we construct the transition sequence $\mathcal{T}=(t_{i(1)},t_{i(2)},...,t_{i(M)})$. According to the construction of $\sigma_0$, we have that $a_1=2q$ and $a_q=2$.


We consider $0\leq l\leq p!-1$. When $|\sigma_{l\cdot(q!+q)}(i_{p}(l+1)+q)-2q|\neq1$, if $\sigma_{l\cdot(q!+q)}(q+1)=2q$, then we let
\begin{equation}
t_{i(j)}=t_{\alpha_{1}(j-l\cdot(q!+q))},\label{Eq 17}
\end{equation}
otherwise we let
\begin{equation}
t_{i(j)}=t_{\alpha_{2}(j-l\cdot(q!+q))} \label{Eq 18}
\end{equation}
for all $l\cdot(q!+q)+1\leq j\leq (l+1)\cdot(q!+q)-1$. When $|\sigma_{l\cdot(q!+q)}(i_{p}(l+1)+q)-2q|=1$, if $\sigma_{l\cdot(q!+q)}(q+1)=2$, then we let
\begin{equation}
t_{i(j)}=t_{\alpha_{1}(j-l\cdot(q!+q))}, \label{Eq 19}
\end{equation}
otherwise we let
\begin{equation}
t_{i(j)}=t_{\alpha_{2}(j-l\cdot(q!+q))} \label{Eq 20}
\end{equation}
for all $l\cdot(q!+q)+1\leq j\leq (l+1)\cdot(q!+q)-1$.

And we let
\begin{equation}
t_{i(j\cdot(q!+q))}=t_{i_{p}(j)+q},~~\text{for all $1\leq j\leq p!$}.\label{Eq 21}
\end{equation}
Hence, $\mathcal{T}$ and $\sigma_0$ can yield one permutation sequence, defined by $C_{\mathcal{T}}=(\sigma_0,\sigma_1,...,\sigma_{M})$. By $(\ref{Eq 17})-(\ref{Eq 21})$ and its first permutation $\sigma_0$, we have that $\sigma_j=t_{i(j)}(\sigma_{j-1})$ for all $1\leq j\leq M$. In the following, we will obtain that $\sigma_0=\sigma_{M}$. Moreover, we let $C_{\mathcal{T}}^{\sigma_0}\triangleq(\sigma_0,\sigma_1,...,\sigma_{M-1})$. Then we will prove that $C_{\mathcal{T}}^{\sigma_0}$ is an $\ell_{\infty}$-snake in the following theorem.

\begin{theorem}
For all $n\geq 6$, there exist an $(n,M,\ell_{\infty})$-snake of size $M=\lceil\frac{n}{2}\rceil! (\lfloor\frac{n}{2}\rfloor+\lfloor\frac{n}{2}\rfloor!)$.\label{Th 1}
\end{theorem}

\begin{proof}
According to the construction of $C_{\mathcal{T}}^{\sigma_0}$, we consider $0\leq l\leq p!-1$. If $|\sigma_{l\cdot(q!+q)}(i_{p}(l+1)+q)-2q|\neq1$ and $\sigma_{l\cdot(q!+q)}(q+1)=2q ~\text{or}~ 2$, by Lemma \ref{Lm 3} and $(\ref{Eq 13})$, we have that
\begin{align}
\sigma_{(l+1)\cdot(q!+q)}=[\sigma_{l\cdot(q!+q)}(i_{p}(l+&1)+q),4,...,2q-2,2,2q,\hat{b}_{2},...,\hat{b}_{p}],
\label{Eq 22}
\end{align}
where $\{\sigma_{l\cdot(q!+q)}(i_{p}(l+1)+q),\hat{b}_2,...,\hat{b}_p\}=O_n$.
Similarly, if $|\sigma_{l\cdot(q!+q)}(i_{p}(l+1)+q)-2q|=1$ and $\sigma_{l\cdot(q!+q)}(q+1)=2q ~\text{or}~ 2$, by Lemma \ref{Lm 3} and $(\ref{Eq 13})$, we can obtain that
\begin{align}
\sigma_{(l+1)\cdot(q!+q)}=[\sigma_{l\cdot(q!+q)}(i_{p}(l+&1)+q),4,...,2q-2,2q,2,\hat{b}_{2},...,\hat{b}_{p}].
\label{Eq 23}
\end{align}
When $|\sigma_{l\cdot(q!+q)}(i_{p}(l+1)+q)-2q|=1$, since $|2q-2|\geq 4$, then we have
\begin{equation}
|\sigma_{l\cdot(q!+q)}(i_{p}(l+1)+q)-2|\geq 2.\label{Eq 24}
\end{equation}
By $(\ref{Eq 22})-(\ref{Eq 24})$, then we have $|\sigma_{(l+1)\cdot(q!+q)}(1)-\sigma_{(l+1)\cdot(q!+q)}(q+1)|\geq 2$ for all $0\leq l\leq p!-1$. Since $\sigma_0(1)=1$ and $\sigma_0(q+1)=2q$, then $|\sigma_0(1)-\sigma_0(q+1)|\geq 2$. Thus, for all $0\leq l\leq p!-1$, $\sigma_{l\cdot(q!+q)}$ satisfies the condition of Lemma \ref{Lm 3}. Hence, by the construction of $C_{\mathcal{T}}^{\sigma_0}$ and Lemma \ref{Lm 3}, for all $0\leq l\leq p!-1,~0\leq j<k\leq q!+q-1$, we have that
\begin{equation}
d_{\infty}(\sigma_{l(q!+q)+j},\sigma_{l(q!+q)+k})\geq 2.\nonumber
\end{equation}

\vskip 1mm Furthermore, for $l,\tilde{l}\in [p!]~\text{and}~l<\tilde{l}$, since the code generated by its transition sequence
$\mathcal{T}_p=(t_{i_{p}(1)},t_{i_{p}(2)},...,t_{i_{p}(p!)})$ is a cyclic and complete $p$-RMGC code, we are assured that for all $0\leq j,\tilde{j}\leq q!+q-1$, the last $p-1$ elements of both $\sigma_{(l-1)(q!+q)+j}$ and $\sigma_{(\tilde{l}-1)(q!+q)+\tilde{j}}$ are all odd and represent two distinct permutations. Hence, we have that
\begin{equation}
d_{\infty}(\sigma_{(l-1)(q!+q)+j},\sigma_{(\tilde{l}-1)(q!+q)+\tilde{j}})\geq 2.\nonumber
\end{equation}
Finally, we will prove that $t_{i(p!(q!+q))}(\sigma_{p!(q!+q)-1})=\sigma_0$. Since the code generated by the transition sequence $\mathcal{T}_p=(t_{i_{p}(1)},t_{i_{p}(2)},\\...,t_{i_{p}(p!)})$ is a cyclic and complete $p$-RMGC code, and by the construction of $\sigma_0$, we have $\sigma_{(p!-1)(q!+q)}(i_{p}(p!-1)+q)=3$. Thus, we have that $|2q-3|\geq 2$ for all $n\geq 6$. Hence, we can obtain that
\begin{align}
\sigma_{(p!-1)(q!+q)}=[3,4,&...,2q-2,2,2q,5,...,2i_{p}(p!)-3,2i_{p}(p!)-1,1,2i_{p}(p!)+1,...,2p-1].\label{Eq 25}
\end{align}
Since $\sigma_{(p!-1)(q!+q)}(i_{p}(p!)+q)=1$ and $\sigma_{(p!-1)(q!+q)}(q+1)=2q$, then we have that
$|\sigma_{(p!-1)(q!+q)}(i_{p}(p!)+q)-2q|\neq 1$. Then, by $(\ref{Eq 22})$ and $(\ref{Eq 25})$, we can obtain that
\begin{align}
\sigma_{p!(q!+q)}=[1,4,...,2q-2,2,2q,3,5,...,2p-1]=\sigma_0\nonumber.
\end{align}
Therefore, $C_{n,\mathcal{T}_{C_n}}$ is an $(n,M,\ell_{\infty})$-snake of size $M=\lceil\frac{n}{2}\rceil! (\lfloor\frac{n}{2}\rfloor+\lfloor\frac{n}{2}\rfloor!)$.
\end{proof}

\subsection{Construction of $\ell_{\infty}$-snakes by using $\mathcal{K}$-snakes}
\label{subsec2}
In this subsection, we will construct $\ell_{\infty}$-snakes by using $\mathcal{K}$-snakes. In order to complete the construction, we need some notations and lemmas of $\mathcal{K}$-snakes.

For a $\mathcal{K}$-snake $C$ over $S_n$, it is a Gray code over $S_n$. Furthermore, for any two distinct codewords $\sigma,\pi\in C$, we have that $d_{\mathcal{K}}(\sigma,\pi)\geq 2$. Here, $d_{\mathcal{K}}(\sigma,\pi)$ is the $\mathcal{K}$-distance between $\sigma$ and $\pi$ in \cite{Jiang2}, where
\begin{equation}
d_{\mathcal{K}}(\sigma,\pi)=|\{(i,j):\sigma^{-1}(i)<\sigma^{-1}(j) \wedge\pi^{-1}(i)>\pi^{-1}(j)\}|.\nonumber
\end{equation}
Moreover, the Kendall's $\tau$-metric is right invariant \cite{Deza}, that is, for every three permutations $\sigma,\pi,\rho\in S_n$, we have $d_{K}(\sigma,\pi)=d_{K}(\sigma\circ\rho,\pi\circ\rho)$. We denote by an $(n,M,\mathcal{K})$-snake a $\mathcal{K}$-snake of size $M$ in $S_n$.

In \cite{Yehezkeally}, Yehezkeally and Schwartz constructed a $(2n+1,M_{2n+1},\mathcal{K})$-snake of length $M_{2n+1}=(2n+1)(2n-1)M_{2n-1}$ in $A_{2n+1}$, from a $(2n-1,M_{2n-1},\mathcal{K})$-snake in $A_{2n-1}$. Later Horovitz and Etzion \cite{Horvitz} improved on this result by constructing a $(2n+1,M_{2n+1},\mathcal{K})$-snake of length $M_{2n+1}=((2n+1)2n-1)M_{2n-1}$ in $A_{2n+1}$, from a $(2n-1,M_{2n-1},\mathcal{K})$-snake in $A_{2n-1}$. The authors in \cite{Horvitz} also presented a direct construction aiming at obtaining a snake in $A_{2n+1}$ of size $\frac{(2n+1)!}{2}-2n+1$. Recently, Zhang and Ge \cite{Zhang} gave a rigorous proof for the Horovitz-Etzion construction of a snake in $A_{2n+1}$ of size $\frac{(2n+1)!}{2}-2n+1$. In order to use the $\mathcal{K}$-snake construction in \cite{Zhang}, we will give the following lemma.



\begin{lemma}\cite[Theorem 4]{Zhang}
\label{LM42}
The construction in \cite{Zhang} can yield a $(2n+1,M_{2n+1},\mathcal{K})$-snake in $A_{2n+1}$ of size $M_{2n+1}=\frac{(2n+1)!}{2}-2n+1$ with the transition sequence including $t_{2n+1}$ for all $n\geq 2$.
\end{lemma}


Furthermore, we require the following lemmas for constructing $\ell_{\infty}$-snakes by using $\mathcal{K}$-snakes.
\begin{lemma}
\label{Lm 5}
Suppose $\{a_j\}_{j=1}^{n}$, $n\geq 2$, is a set of integers of the same parity. Let $\sigma_{i}=[\sigma_{i}(1),..,\sigma_{i}(n),\sigma_{i}(n+1),b_{n+2},...,b_{m}]\in S_m$ for $i=1,2$, where $\sigma_1\neq\sigma_2,$ $\{\sigma_{i}(j)\}_{j=1}^{n+1}=\{a_j\}_{j=1}^{n}\cup\{x\}$ for $i=1,2$, and the parity of $x$ differs from that of the elements of $\{a_j\}_{j=1}^{n}$. If $\sigma_1$ and $\sigma_2$ are both odd permutations or even permutations, then $d_{\infty}(\sigma_1,\sigma_2)\geq2$.
\end{lemma}

\begin{proof}
Since $\sigma_1\neq \sigma_2$, then $d_{\infty}(\sigma_1,\sigma_2)\geq1$. Suppose $d_{\infty}(\sigma_1,\sigma_2)<2$, we have that $d_{\infty}(\sigma_1,\sigma_2)=1$. Without loss of generality, we let $\sigma_1=[a_1,a_2,..,a_n,x,b_{n+2},..,b_m]$, $|a_{j_1}-x|=1$, and $|a_{j_2}-x|=1$, where $j_1,j_2\in[n]$. Since $\{a_j\}_{j=1}^{n}$ are the same parity and $d_{\infty}(\sigma_1,\sigma_2)=1$, then
$\sigma_2=[a_1,...,a_{j_1-1},x,a_{j_1+1},...,a_{n},a_{j_1},b_{n+2},...,b_m]$ or $\sigma_2=[a_1,...,a_{j_2-1},x,a_{j_2+1},...,a_{n},a_{j_2},b_{n+2},...,b_m]$. Thus, $\sigma_2$ can be obtained from $\sigma_1$ using one transposition of $a_{j_i}$ and $x$ for $i=1~\text{or}~2$. Hence, the parity of $\sigma_1$ differs from the parity of $\sigma_2$, which causes a contradiction. Then, we have that $d_{\infty}(\sigma_1,\sigma_2)\geq2$.
\end{proof}

\begin{lemma}
\label{Lm 6}
Suppose $C_n$ is an $(n,M_n,\mathcal{K})$-snake in $A_n$ with its first permutation $\pi_0$ and one transition sequence $\mathcal{T}_{C_n}=(t_{\hat{i}(1)},t_{\hat{i}(2)},...,t_{\hat{i}(M_n)})$. For any $\sigma_0\in S_n$, by applying the transition sequence $\mathcal{T}_{C_n}$ on the permutation $\sigma_0$, then we can obtain another $(n,M_n,\mathcal{K})$-snake, denoted by $\hat{C}_n=(\sigma_0,\sigma_1,...,\sigma_{M_n-1})$, where $\sigma_j=t_{\hat{i}(j)}(\sigma_{j-1})$ for all $j\in [M_n-1]$. Moreover, the parities of all the permutations of $\hat{C}_n$ are same.
\end{lemma}

\begin{proof}
According to the \cite[Lemma 3]{Wang}, we have that $\hat{C}_n$ is an $(n,M_n,\mathcal{K})$-snake. Suppose $C_n=(\pi_0,\pi_1,...,\pi_{M_n-1})$. According to the \cite[Lemma 2]{Wang}, we can obtain that
\begin{equation}
\sigma_j\circ{\sigma_0}^{-1}=\pi_j\circ{\pi_0}^{-1}~\text{for all $j \in [M_n-1]$}.\label{Eq 26}
\end{equation}
Since the Kendall's $\tau$-metric is right invariant, and by $(\ref{Eq 26})$, then for any two distinct permutations $\sigma_j,\sigma_k \in \hat{C}_n$, we have that
\begin{align}
d_{\mathcal{K}}(\sigma_j,\sigma_k)=&d_{\mathcal{K}}(\sigma_j\circ({\sigma_0}^{-1}\circ\pi_0),\sigma_k\circ({\sigma_0}^{-1}\circ
\pi_0))\nonumber\\
=&d_{\mathcal{K}}(\sigma_j\circ{\sigma_0}^{-1},\sigma_k\circ{\sigma_0}^{-1})\nonumber\\
=&d_{\mathcal{K}}(\pi_j\circ{\pi_0}^{-1},\pi_k\circ{\pi_0}^{-1})\nonumber\\
=&d_{\mathcal{K}}(\pi_j,\pi_k).\label{Eq 27}
\end{align}
Furthermore, since $C_n\subset A_n$, then for two distinct permutations $\pi_j,\pi_k$, we can obtain that
\begin{equation}
d_{\mathcal{K}}(\pi_j,\pi_k)=2s\label{Eq 28}
\end{equation}
where $s$ is a positive integer. By $(\ref{Eq 27})$ and $(\ref{Eq 28})$, we can obtain that the parities of all the permutations of $\hat{C}_n$ are same.
\end{proof}

The following lemma gives the construction of a basic block which is useful for the construction of $\ell_{\infty}$-snakes by using $\mathcal{K}$-snakes.

\begin{lemma}
\label{Lm 7}
Let $\{a_{j}\}_{j=1}^{k}$ be a set of integers of the same parity, and let $\{b_{j}\}_{j=1}^{l}$ be also a set of integers of the same parity such that $\{a_{j}\}_{j=1}^{k}\cup\{b_{j}\}_{j=1}^{l}=[n]$. And we let $\sigma\triangleq[a_1,b_1,b_2,b_3,...,b_l,a_2,a_3...,a_k]$. Suppose we have an $(l+1,M_{l+1},\mathcal{K})$-snake in $A_{l+1}$ with one transition sequence $\mathcal{T}_{\mathcal{K},l+1}=(t_{\hat{i}(1)},t_{\hat{i}(2)}, ...,t_{\hat{i}(M_{l+1})})$ such that $t_{\hat{i}(M_{l+1})}=t_{l+1}$. Then, there exists a noncyclic $(n,M_{l+1},\ell_{\infty})$-snake starting with $\sigma$ and ending with the permutation $\pi=[b_1,b_2,..,b_l,a_1,a_2,...,a_k]$.
\end{lemma}

\begin{proof}
We let $C_{\hat{\mathcal{T}}_{l+1}}^{\sigma,\pi}$ be the claimed noncyclic $\ell_{\infty}$-snake, where $C_{\hat{\mathcal{T}}_{l+1}}^{\sigma,\pi}=(\sigma_0,\sigma_1,...,\sigma_{M_{l+1}-1})$ and $\hat{\mathcal{T}}_{l+1}=(t_{\alpha(1)},t_{\alpha(2)},...,\\ t_{\alpha(M_{l+1}-1)})$.

Firstly, we denote by $\sigma_0\triangleq\sigma$. Next, we construct the transition sequence $\hat{\mathcal{T}}_{l+1}$. We let
\begin{equation}
t_{\alpha(j)}=t_{\hat{i}(j)}~\text{for all $j\in [M_{l+1}-1]$}.\label{Eq 29}
\end{equation}
By $(\ref{Eq 29})$ and its first permutation $\sigma_0$, we have that
\begin{equation}
\sigma_j=[\sigma_j(1),...,\sigma_{j}(l+1),a_2,a_3,...,a_k]\nonumber
\end{equation}
for all $j\in [M_{l+1}-1]$. By $(\ref{Eq 29})$ and Lemma \ref{Lm 6}, due to the $(l+1,M_{l+1},\mathcal{K})$-snake in $A_{l+1}$, we have that $C_{\hat{\mathcal{T}}_{l+1}}^{\sigma,\pi}$ is a noncyclic Gray code and the parities of all the permutations of $C_{\hat{\mathcal{T}}_{l+1}}^{\sigma,\pi}$ are same. Since $t_{\hat{i}(M_{l+1})}=t_{l+1}$, then we have $\pi=\sigma_{M_{l+1}-1}=[b_1,b_2,...,b_l,a_1,a_2,...,a_k]$.

Finally, for any two distinct permutations $\sigma_{j_1},\sigma_{j_2}\in C_{\hat{\mathcal{T}}_{l+1}}^{\sigma,\pi}$, since their parities are same and $\sigma_{j_i}=[\sigma_{j_i}(1),...,\sigma_{j_i}(l+1),a_2,a_3,...,a_k]$, for $i=1~\text{or}~2$, by Lemma \ref{Lm 5}, we have that
\begin{equation}
d_{\infty}(\sigma_{j_1},\sigma_{j_2})\geq 2.\nonumber
\end{equation}
Hence, we can obtain that $C_{\hat{\mathcal{T}}_{l+1}}^{\sigma,\pi}$ is a noncyclic $(n,M_{l+1},\ell_{\infty})$-snake starting with $\sigma$ and ending with the permutation $\pi=[b_1,b_2,..,b_l,a_1,a_2,...,a_k]$.
\end{proof}

When $n=4k+1$, $k\geq 1$, then $[n]$ has $2k$ even elements and $2k+1$ odd ones. In the following, by Lemma \ref{Lm 7}, we will give one construction of an $(n,M,\ell_{\infty})$-snake by using some $\mathcal{K}$-snakes. Firstly, we denote by $\sigma_0$ an initial permutation, where
\begin{equation}
\sigma_0=[1,2,4,...,4k,3,5...,4k+1].\nonumber
\end{equation}
And we construct a transition sequence, denoted by $\mathcal{T}_C=(t_{i(1)},t_{i(2)},...,t_{i(M)})$. By the transition sequence $\mathcal{T}_C$ and the initial permutation $\sigma_0$, we can get a permutation sequence, denoted by $C_{\mathcal{T}_C}^{\sigma_0}=(\sigma_0,\sigma_1,...,\sigma_{M-1})$. Given a $(2k+1,M_{2k+1},\mathcal{K})$-snake in $A_{2k+1}$ with one transition sequence $(t_{\alpha(1)},t_{\alpha(2)},...,t_{\alpha(M_{2k+1})})$ and $t_{\alpha(M_{2k+1})}=t_{2k+1}$, by Lemma \ref{Lm 7}, we take a noncyclic $(n,M_{2k+1},\ell_{\infty})$-snake by using the following transition sequence
\begin{equation}
\hat{\mathcal{T}}_{2k+1}=(t_{\alpha(1)},t_{\alpha(2)},...,t_{\alpha(M_{2k+1}-1)}).\label{Eq 30}
\end{equation}
Moreover, by Lemma \ref{Lm 1}, we can obtain a cyclic and complete $(2k+1)$-RMGC by using the following transition sequence
\begin{equation}
\mathcal{T}_{2k+1}=(t_{i_{2k+1}(1)},t_{i_{2k+1}(2)},...,t_{i_{2k+1}((2k+1)!)}).\label{Eq 31}
\end{equation}
By $(\ref{Eq 30})-(\ref{Eq 31})$, we construct the transition sequence $\mathcal{T}_C=(t_{i(1)},t_{i(2)},...,t_{i(M)})$ such that $M=M_{2k+1}(2k+1)!$. We let
\begin{equation}
t_{i(j)}=t_{{\alpha}(j-l\cdot M_{2k+1})}\label{Eq 32}
\end{equation}
for all $l\cdot M_{2k+1}+1\leq j\leq (l+1)\cdot M_{2k+1}-1$ and $0\leq l\leq(2k+1)!-1$, and we let
\begin{equation}
t_{i(j\cdot M_{2k+1})}=t_{i_{2k+1}(j)+2k} \label{Eq 33}
\end{equation}
for all $1\leq j\leq (2k+1)!$.

Secondly, by $(\ref{Eq 32})-(\ref{Eq 33})$ and the initial permutation $\sigma_0$, we obtain the permutation sequence $\sigma_j=t_{i(j)}(\sigma_{j-1})$ for all $1\leq j\leq M_{2k+1}(2k+1)!-1$.

Finally, in the following theorem, we will prove that $C_{\mathcal{T}_C}^{\sigma_0}$ is a $(4k+1,M_{2k+1}(2k+1)!,\ell_{\infty})$-snake.

Similarly, when $n=4k+3, k\geq1$, then $[n]$ has $2k+1$ even elements and $2k+2$ odd ones. In the following, by Lemma \ref{Lm 7}, we will give another construction of an $(n,\hat{M},\ell_{\infty})$-snake by using some $\mathcal{K}$-snakes. Firstly, we denote by $\hat{\sigma}_0$ an initial permutation, where
\begin{equation}
\hat{\sigma}_0=[2,1,3,5,...,4k+3,4,6...,4k+2].\label{Eq 40}
\end{equation}
And we construct another transition sequence, denoted by $\mathcal{T}_{\hat{C}}=(t_{\hat{i}(1)},t_{\hat{i}(2)},...,t_{\hat{i}(\hat{M})})$. By the transition sequence $\mathcal{T}_{\hat{C}}$ and the initial permutation $\hat{\sigma}_0$, we can get a permutation sequence, denoted by $\hat{C}_{\mathcal{T}_{\hat{C}}}^{\hat{\sigma}_0}=(\hat{\sigma}_0,\hat{\sigma}_1,...,\hat{\sigma}_{\hat{M}-1})$.

Given a $(2k+3,M_{2k+3},\mathcal{K})$-snake in $A_{2k+3}$ with one transition sequence $(t_{\beta(1)},t_{\beta(2)},...,t_{\beta(M_{2k+3})})$ and $t_{\beta(M_{2k+3})}=t_{2k+3}$, by Lemma \ref{Lm 7}, we take a noncyclic $(4k+3,M_{2k+3},\ell_{\infty})$-snake by using the following transition sequence
\begin{equation}
\hat{\mathcal{T}}_{2k+3}=(t_{\beta(1)},t_{\beta(2)},...,t_{\beta(M_{2k+3}-1)}).\label{Eq 34}
\end{equation}
Moreover, by Lemma \ref{Lm 1}, we can obtain a cyclic and complete $(2k+1)$-RMGC by using the following transition sequence
\begin{equation}
\mathcal{T}_{2k+1}=(t_{i_{2k+1}(1)},t_{i_{2k+1}(2)},...,t_{i_{2k+1}((2k+1)!)}).\label{Eq 35}
\end{equation}
By $(\ref{Eq 34})-(\ref{Eq 35})$, we construct the transition sequence $\mathcal{T}_{\hat{C}}=(t_{\hat{i}(1)},t_{\hat{i}(2)},...,t_{\hat{i}(\hat{M})})$ such that $\hat{M}=M_{2k+3}(2k+1)!$. We let
\begin{equation}
t_{\hat{i}(j)}=t_{\beta(j-l\cdot M_{2k+3})} \label{Eq 36}
\end{equation}
for all $l\cdot M_{2k+3}+1\leq j\leq (l+1)\cdot M_{2k+3}-1$ and $0\leq l\leq(2k+1)!-1$, and we let
\begin{equation}
t_{\hat{i}(j\cdot M_{2k+3})}=t_{i_{2k+1}(j)+2k+2}\label{Eq 37}
\end{equation}
for all $1\leq j\leq (2k+1)!$.

Secondly, by $(\ref{Eq 36})-(\ref{Eq 37})$ and its first permutation $\hat{\sigma}_0$, we obtain the permutation sequence $\hat{\sigma}_j=t_{\hat{i}(j)}(\hat{\sigma}_{j-1})$ for all $1\leq j\leq M_{2k+3}(2k+1)!-1$.

Finally, in the following theorem, we will also prove that $\hat{C}_{\mathcal{T}_{\hat{C}}}^{\hat{\sigma}_0}$ is a $(4k+3,M_{2k+3}(2k+1)!,\ell_{\infty})$-snake.

\begin{theorem}
\label{Th 2}
When $n=4k+1$ and $k\geq 1$, given a $(2k+1,M_{2k+1},\mathcal{K})$-snake in $A_{2k+1}$, there exists an $(n,M,\ell_{\infty})$-snake of size $M=M_{2k+1}\cdot(2k+1)!$. Moreover, when $n=4k+3$ and $k\geq 1$, given a $(2k+3,M_{2k+3},\mathcal{K})$-snake in $A_{2k+3}$, there exists an $(n,\hat{M},\ell_{\infty})$-snake of size $\hat{M}=M_{2k+3}\cdot(2k+1)!$.
\end{theorem}

\begin{proof}
When $n=4k+1$, we will prove that the above $C_{\mathcal{T}_C}^{\sigma_0}$ is an $\ell_{\infty}$-snake. Since $\sigma_0=[1,2,4,...,4k,3,5,...,4k+1]$, by the construction of this $\ell_{\infty}$-snake, we have that for all $0\leq l\leq (2k+1)!-1$, $\sigma_{l\cdot M_{2k+1}}$ satisfies the condition of Lemma \ref{Lm 7}. Hence, by the construction of $C_{\mathcal{T}_C}^{\sigma_0}$ and Lemma \ref{Lm 7}, for all $0\leq l\leq (2k+1)!-1~\text{and}~0\leq j<k\leq M_{2k+1}-1$, we have
\begin{equation}
d_{\infty}(\sigma_{l\cdot M_{2k+1}+j},\sigma_{l\cdot M_{2k+1}+k})\geq 2.\nonumber
\end{equation}

Furthermore, for $l,\tilde{l}\in [(2k+1)!]~\text{and}~l<\tilde{l}$, since the code generated by the transition sequence $\mathcal{T}_{2k+1}=(t_{i_{2k+1}(1)},t_{i_{2k+1}(2)},...,\\ t_{i_{2k+1}((2k+1)!)})$ is a cyclic and complete $(2k+1)$-RMGC code, we are assured that for all $0\leq j,\tilde{j}\leq M_{2k+1}-1$, the last $2k$ elements of both $\sigma_{(l-1)M_{2k+1}+j}$ and $\sigma_{(\tilde{l}-1)M_{2k+1}+\tilde{j}}$ are all odd and represent two distinct permutations. Hence, we have that
\begin{equation}
d_{\infty}(\sigma_{(l-1)M_{2k+1}+j},\sigma_{(\tilde{l}-1)M_{2k+1}+\tilde{j}})\geq 2.\nonumber
\end{equation}

Finally, we note that $t_{i(M_{2k+1}(2k+1)!)}(\sigma_{M_{2k+1}(2k+1)!-1})=\sigma_0$, since the code generated by the transition sequence $\mathcal{T}_{2k+1}$ is cyclic. Therefore, $C_{\mathcal{T}_C}^{\sigma_0}$ is an $(n,M,\ell_{\infty})$-snake of size $M=M_{2k+1}(2k+1)!$.

Similarly, when $n=4k+3$, according to the construction of $\hat{C}_{\mathcal{T}_{\hat{C}}}^{\hat{\sigma}_0}$, we can obtain that $\hat{C}_{\mathcal{T}_{\hat{C}}}^{\hat{\sigma}_0}$ is a $(4k+3,\hat{M},\ell_{\infty})$-snake of size $\hat{M}=M_{2k+3}(2k+1)!$.
\end{proof}

\begin{corollary}
\label{Cor1}
If $n=4k+1$ and $k\geq 2$, there exists an $(n,M,\ell_{\infty})$-snake of size $M=\big(\frac{(2k+1)!}{2}-2k+1\big)\cdot(2k+1)!$.
Moreover, if $n=4k-1$ and $k\geq 2$, there also exists an $(n,\hat{M},\ell_{\infty})$-snake of size $\hat{M}=\big(\frac{(2k+1)!}{2}-2k+1\big)\cdot(2k-1)!$
\end{corollary}
\begin{proof}
According to Theorem \ref{Th 2} and Lemma \ref{LM42}, we can prove this corollary.
\end{proof}

\section{Examples of $\ell_{\infty}$-snakes}
\label{sec4}
\subsection{One example of $\ell_{\infty}$-snakes by using cyclic and complete RMGCs}
In this subsection, we give an example of $\ell_{\infty}$-snakes which is constructed by using cyclic and complete RMGCs. Consider $n=6$, then we have that $p=q=3$. By Lemma \ref{Lm 3}, firstly we will construct two kinds of noncyclic $\ell_{\infty}$-snakes which are basic building blocks for $\ell_{\infty}$-snakes. Secondly we will give a cyclic $\ell_{\infty}$-snake.

Now, we will start this example with an initial permutation, denoted by $\sigma_0$. By $(\ref{Eq 13})$, we have that
\begin{equation}
\sigma_0=[1,4,2,6,3,5].\nonumber
\end{equation}
In order to construct the blocks, we need one transition sequence of a cyclic and complete $3$-RMGC, i.e, $\mathcal{T}_3=(t_3,t_3,t_2,t_3,t_3,t_2)$. By Lemma \ref{Lm 3}, we can obtain two transition sequences $\mathcal{T}_{C}$ and $\mathcal{T}_{\hat{C}}$, where
\begin{equation}
\mathcal{T}_{C}=(t_3,t_3,t_4,t_2,t_3,t_3,t_2,t_3)\nonumber
\end{equation}
 and
\begin{equation}
\mathcal{T}_{\hat{C}}=(t_3,t_3,t_4,t_3,t_3,t_2,t_3,t_3).\nonumber
\end{equation}
Next, we will give two noncyclic $(6,3!+3,\ell_{\infty})$-snakes by the two transition sequences and $\sigma_0$. Then, one noncyclic $(6,3!+3,\ell_{\infty})$-snake is constructed by $\mathcal{T}_{C}$ and $\sigma_0$, which is depicted by Figure 1 as follows
\begin{figure}[H]
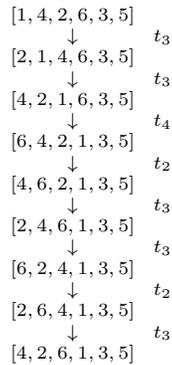

\label{Fig 1}
\addtolength{\tabcolsep}{-5pt}
\captionsetup{{font}={scriptsize}}
\scriptsize

\begin{center}
\begin{tabular}{cc}
$[1,4,2,6,3,5]$& \\
$\downarrow$&$~~t_3$\\
$[2,1,4,6,3,5]$&\\
$\downarrow$&$~~t_3$\\
$[4,2,1,6,3,5]$&\\
$\downarrow$&$~~t_4$\\
$[6,4,2,1,3,5]$&\\
$\downarrow$&$~~t_2$\\
$[4,6,2,1,3,5]$&\\
$\downarrow$&$~~t_3$\\
$[2,4,6,1,3,5]$&\\
$\downarrow$&$~~t_3$\\
$[6,2,4,1,3,5]$&\\
$\downarrow$&$~~t_2$\\
$[2,6,4,1,3,5]$&\\
$\downarrow$&$~~t_3$\\
$[4,2,6,1,3,5]$&
\end{tabular}
\end{center}
\caption*{Fig. 1.~ A noncyclic $(6,3!+3,\ell_{\infty})$-snake constructed by $\mathcal{T}_{C}$ and $\sigma_0$.}
\end{figure}

\noindent
Another noncyclic $(6,3!+3,\ell_{\infty})$-snake is constructed by $\mathcal{T}_{\hat{C}}$ and $\sigma_0$, which is depicted by Figure 2 as follows

\begin{figure}[H]
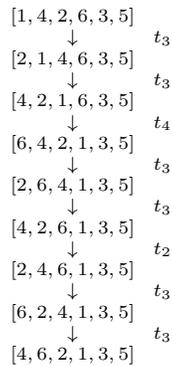

\addtolength{\tabcolsep}{-5pt}
\captionsetup{{font}={scriptsize}}
\scriptsize

\begin{center}
\begin{tabular}{cc}
$[1,4,2,6,3,5]$& \\
$\downarrow$&$~~t_3$\\
$[2,1,4,6,3,5]$&\\
$\downarrow$&$~~t_3$\\
$[4,2,1,6,3,5]$&\\
$\downarrow$&$~~t_4$\\
$[6,4,2,1,3,5]$&\\
$\downarrow$&$~~t_3$\\
$[2,6,4,1,3,5]$&\\
$\downarrow$&$~~t_3$\\
$[4,2,6,1,3,5]$&\\
$\downarrow$&$~~t_2$\\
$[2,4,6,1,3,5]$&\\
$\downarrow$&$~~t_3$\\
$[6,2,4,1,3,5]$&\\
$\downarrow$&$~~t_3$\\
$[4,6,2,1,3,5]$&
\end{tabular}
\end{center}
\caption*{Fig. 2.~ A noncyclic $(6,3!+3,\ell_{\infty})$-snake constructed by $\mathcal{T}_{\hat{C}}$ and $\sigma_0$.}
\label{Fig 2}
\end{figure}
\noindent
By the two kinds of basic noncyclic $\ell_{\infty}$-snakes, we have a cyclic $(6,3!(3!+3),\ell_{\infty})$-snake described in Figure 3 as follows
\begin{figure}[H]
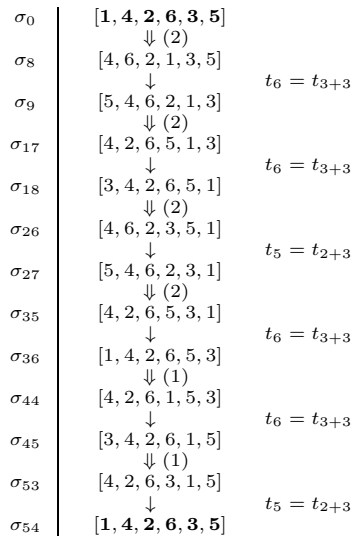


\captionsetup{{font}={scriptsize}}
\scriptsize
\begin{center}
\begin{tabular}{c|cc}
$\sigma_{0}$ &$~~~\bf{[1,4,2,6,3,5]}$& \\
$~$           &$~~~~\Downarrow(2)$& \\
$\sigma_{8}$ &$~~~[4,6,2,1,3,5]$& \\
$~$           &$\downarrow$& $~t_6=t_{3+3}$\\
$\sigma_{9}$ &$~~~[5,4,6,2,1,3]$& \\
$~$           &$~~~~\Downarrow(2)$& \\
$\sigma_{17}$&$~~~[4,2,6,5,1,3]$& \\
$~$           &$\downarrow$& $~t_6=t_{3+3}$\\
$\sigma_{18}$&$~~~[3,4,2,6,5,1]$& \\
$~$           &$~~~~\Downarrow(2)$& \\
$\sigma_{26}$&$~~~[4,6,2,3,5,1]$& \\
$~$           &$\downarrow$& $~t_5=t_{2+3}$\\
$\sigma_{27}$&$~~~[5,4,6,2,3,1]$& \\
$~$           &$~~~~\Downarrow(2)$& \\
$\sigma_{35}$&$~~~[4,2,6,5,3,1]$& \\
$~$           &$\downarrow$& $~t_{6}=t_{3+3}$\\
$\sigma_{36}$&$~~~[1,4,2,6,5,3]$& \\
$~$           &$~~~~\Downarrow(1)$& \\
$\sigma_{44}$&$~~~[4,2,6,1,5,3]$& \\
$~$           &$\downarrow$& $~t_{6}=t_{3+3}$\\
$\sigma_{45}$&$~~~[3,4,2,6,1,5]$& \\
$~$           &$~~~~\Downarrow(1)$& \\
$\sigma_{53}$&$~~~[4,2,6,3,1,5]$& \\
$~$           &$\downarrow$&  $~t_{5}=t_{2+3}$\\
$\sigma_{54}$&$~~~\bf{[1,4,2,6,3,5]}$&
\end{tabular}
\end{center}
\caption*{Fig. 3.~ A $(6,54,\ell_{\infty})$-snake obtained by using a cyclic and complete $3$-RMGC.}
\label{Fig 3}
\end{figure}
\noindent
In Figure 3, ``$\Downarrow(1)$'' stands for an omitted transition sequence $\mathcal{T}_{C}=(t_3,t_3,t_4,t_2,t_3,t_3,t_2,t_3)$. While ``$\Downarrow(2)$'' stands for another omitted transition sequence $\mathcal{T}_{\hat{C}}=(t_3,t_3,t_4,t_3,t_3,t_2,t_3,t_3).$ Hence, when $n=6$, by using one cyclic and complete $3$-RMGC, we can construct a cyclic $\ell_{\infty}$-snake of size $54$.

\subsection{One example of $\ell_{\infty}$-snakes by using $\mathcal{K}$-snakes}
In this subsection, we present an example of $\ell_{\infty}$-snakes constructed by using $\mathcal{K}$-snakes. Consider $n=7$, then we have $4$ odd elements and $3$ even ones in $\{1,2,3,4,5,6,7\}$. Moreover, Horovitz and Etzion \cite{Horvitz} gave a $(5,57,\mathcal{K})$-snake in $A_5$ with one transition sequence, denoted by $\mathcal{T}_{\mathcal{K},5}=(\hat{\mathcal{T}},\hat{\mathcal{T}},\hat{\mathcal{T}})$, where $\hat{\mathcal{T}}$ is a partial transition sequence of $\mathcal{T}_{\mathcal{K},5}$ and  $\hat{\mathcal{T}}=(t_3,t_3,t_5,t_3,t_3,t_5,t_3,t_5,t_5,t_3,t_3,t_5,t_3,t_3,t_5,t_3,t_5,t_5,t_5)$.

By Theorem \ref{Th 2}, given this $(5,M_5,\mathcal{K})$-snake of size $57$ in \cite{Horvitz}, we can obtain a $(7,M,\ell_{\infty})$-snake of size $M=342$.

Firstly, we will start this example with one initial permutation, denoted by $\hat{\sigma}_0$. By $(\ref{Eq 40})$, we have that
\begin{equation}
\hat{\sigma}_0=[2,1,3,5,7,4,6].\nonumber
\end{equation}
By Lemma \ref{Lm 7} and $\mathcal{T}_{\mathcal{K},5}$, we can construct one transition sequence, denoted by $\hat{\mathcal{T}}_{\mathcal{K},5}$, where
\begin{align*}
\hat{\mathcal{T}}_{\mathcal{K},5}=(\hat{\mathcal{T}},\hat{\mathcal{T}},t_3,t_3,t_5,t_3,t_3,t_5,t_3,
t_5,t_5,t_3,t_3,t_5,t_3,t_3,t_5,t_3,t_5,t_5).
\end{align*}
Hence, we will give one noncyclic $(7,57,\ell_{\infty})$-snake by the transition sequence $\hat{\mathcal{T}}_{\mathcal{K},5}$ and $\hat{\sigma}_0$. Then, the noncyclic $(7,57,\ell_{\infty})$-snake determined by $\hat{\mathcal{T}}_{\mathcal{K},5}$ and $\hat{\sigma}_0$ is depicted by
Figure 4 as follows
\begin{figure}[H]
\addtolength{\tabcolsep}{-5pt}
\captionsetup{{font}={scriptsize}}
\scriptsize
\begin{center}
\begin{tabular}{c|c|c|c|c|c|c|c|c|c|c|c|c|c|c|c|c|c|c|c|c|c|c|c|c|c|c|c|c|c|c|c|c|c|c|c|c|c|c|c|c|c|c|c|c|c|c|c|c|c|c|c|c|c|c|c|c}
2&3&1&7&3&1&5&3&2&7&3&2&1&3&2&5&3&7&1&2&7&1&5&7&1&3&7&2&5&7&2&1&7&2&3&7&5&1&2&5&1&3&5&1&7&5&2&3&5&2&1&5&2&7&5&3&1\\
1&2&3&1&7&3&1&5&3&2&7&3&2&1&3&2&5&3&7&1&2&7&1&5&7&1&3&7&2&5&7&2&1&7&2&3&7&5&1&2&5&1&3&5&1&7&5&2&3&5&2&1&5&2&7&5&3\\
3&1&2&3&1&7&3&1&5&3&2&7&3&2&1&3&2&5&3&7&1&2&7&1&5&7&1&3&7&2&5&7&2&1&7&2&3&7&5&1&2&5&1&3&5&1&7&5&2&3&5&2&1&5&2&7&5\\
5&5&5&2&2&2&7&7&1&5&5&5&7&7&7&1&1&2&5&3&3&3&2&2&2&5&5&1&3&3&3&5&5&5&1&1&2&3&7&7&7&2&2&2&3&3&1&7&7&7&3&3&3&1&1&2&7\\
7&7&7&5&5&5&2&2&7&1&1&1&5&5&5&7&7&1&2&5&5&5&3&3&3&2&2&5&1&1&1&3&3&3&5&5&1&2&3&3&3&7&7&7&2&2&3&1&1&1&7&7&7&3&3&1&2\\
4&4&4&4&4&4&4&4&4&4&4&4&4&4&4&4&4&4&4&4&4&4&4&4&4&4&4&4&4&4&4&4&4&4&4&4&4&4&4&4&4&4&4&4&4&4&4&4&4&4&4&4&4&4&4&4&4\\
6&6&6&6&6&6&6&6&6&6&6&6&6&6&6&6&6&6&6&6&6&6&6&6&6&6&6&6&6&6&6&6&6&6&6&6&6&6&6&6&6&6&6&6&6&6&6&6&6&6&6&6&6&6&6&6&6\\
\end{tabular}
\end{center}

\caption*{Fig. 4.~ A noncyclic $(7,57,\ell_{\infty})$-snake constructed by  $\hat{\mathcal{T}}_{\mathcal{K},5}$ and $\hat{\sigma}_0$.}
\label{Fig 4}
\end{figure}
Here, every column in Figure 4 represents one permutation over $\{1,2,3,4,5,6,7\}$. Moreover, we have a cyclic and complete $3$-RMGC with the transition sequence $\mathcal{T}_3=(t_3,t_3,t_2,t_3,t_3,t_2)$. By Figure 4 and $\mathcal{T}_3$, we can obtain a $(7,342,\ell_{\infty})$-snake depicted by Figure 5 as follows





\begin{figure}[H]

\captionsetup{{font}={scriptsize}}
\scriptsize
\begin{center}
\begin{tabular}{c|cc}
$\sigma_{0}$ &$~~~\bf{[2,1,3,5,7,4,6]}$& \\
$~$           &$\Downarrow$& \\
$\sigma_{56}$ &$~~~[1,3,5,7,2,4,6]$& \\
$~$           &$\downarrow$& $~t_7=t_{3+4}$\\
$\sigma_{57}$ &$~~~[6,1,3,5,7,2,4]$& \\
$~$           &$\Downarrow$& \\
$\sigma_{113}$&$~~~[1,3,5,7,6,2,4]$& \\
$~$           &$\downarrow$& $~t_7=t_{3+4}$\\
$\sigma_{114}$&$~~~[4,1,3,5,7,6,2]$& \\
$~$           &$\Downarrow$& \\
$\sigma_{170}$&$~~~[1,3,5,7,4,6,2]$& \\
$~$           &$\downarrow$& $~t_6=t_{2+4}$\\
$\sigma_{171}$&$~~~[6,1,3,5,7,4,2]$& \\
$~$           &$\Downarrow$& \\
$\sigma_{227}$&$~~~[1,3,5,7,6,4,2]$& \\
$~$           &$\downarrow$& $~t_{7}=t_{3+4}$\\
$\sigma_{228}$&$~~~[2,1,3,5,7,6,4]$& \\
$~$           &$\Downarrow$& \\
$\sigma_{284}$&$~~~[1,3,5,7,2,6,4]$& \\
$~$           &$\downarrow$& $~t_{7}=t_{3+4}$\\
$\sigma_{285}$&$~~~[4,1,3,5,7,2,6]$& \\
$~$           &$\Downarrow$& \\
$\sigma_{341}$&$~~~[1,3,5,7,4,2,6]$& \\
$~$           &$\downarrow$&  $~t_{6}=t_{2+4}$\\
$\sigma_{342}$&$~~~\bf{[2,1,3,5,7,4,6]}$&
\end{tabular}
\end{center}
\caption*{Fig. 5.~ A $(7,342,\ell_{\infty})$-snake constructed by using a $\mathcal{K}$-snake in $A_5$.}
\label{Fig 5}
\end{figure}
In Figure 5, ``$\Downarrow$'' stands for an omitted transition sequence ${\hat{\mathcal{T}}}_{\mathcal{K},5}$. Therefore, when $n=7$, by using $\mathcal{K}$-snakes in $A_n$, we can obtain a cyclic $(7,342,\ell_{\infty})$-snake.

\section{Comparison}
\label{sec5}
Yehezkeally and Schwartz \cite{Yehezkeally} presented one construction of an $(n,M_{n,0},\ell_{\infty})$-snake of size
\begin{equation}
M_{n,0}=\lceil\frac{n}{2}\rceil!(\lfloor\frac{n}{2}\rfloor+(\lfloor\frac{n}{2}\rfloor-1)!)
~\text{for all $n\geq 4$},\label{Eq 39}
\end{equation}
which is far less than the upper bound $\frac{n!}{2^{ \lfloor\frac{n}{2}\rfloor}}$ on the length of $\ell_{\infty}$-snakes in $S_n$.

Based on their construction of $\ell_{\infty}$-snakes, we proposed one construction of $\ell_{\infty}$-snakes by using cyclic and complete RMGCs. In this construction, we could obtain an $(n,M_{n,1},\ell_{\infty})$-snake of size
\begin{equation}
M_{n,1}=\lceil\frac{n}{2}\rceil!(\lfloor\frac{n}{2}\rfloor+(\lfloor\frac{n}{2}\rfloor)!)~\text{ for all $n\geq 6$}.\label{Eq 41}
\end{equation}
Hence, our $\ell_{\infty}$-snakes are better than Yehezkeally and Schwartz's ones for all $n\geq 6$.

Moreover, we also gave another construction of $\ell_{\infty}$-snakes by using $\mathcal{K}$-snakes. By Corollary \ref{Cor1}, we can obtain an $(n,M_{n,2},\ell_{\infty})$-snake, where
\begin{equation}
M_{n,2}=
\begin{cases}
\big(\frac{(2k+1)!}{2}-2k+1\big)\cdot (2k+1)!~~  \text{if $n=4k+1$},\\
\big(\frac{(2k+1)!}{2}-2k+1\big)\cdot (2k-1)! ~~ \text{if $n=4k-1$},
\end{cases}\label{Eq 42}
\end{equation}
for all $k\geq 2$.

By $(\ref{Eq 41})$ and $(\ref{Eq 42})$, when $n=4k+1~\text{or}~4k-1$, and $k\geq 2$, we have that $M_{n,2}> M_{n,1}$. Thus, we can obtain that
\begin{equation}
M_{n,2}> M_{n,1}>M_{n,0}\label{Eq 44}
\end{equation}
for all $n=4k+1~\text{or}~4k-1$, and $k\geq 2$. Hence, by $(\ref{Eq 44})$, the second construction is superior to the first one and Yehezkeally and Schwartz's one in some cases. For example, when $n=7$, by $(\ref{Eq 39})-(\ref{Eq 42})$, we have that $M_{7,0}=120,M_{7,1}=216,M_{7,2}=342$. Therefore, we can obtain that $M_{7,2}>M_{7,1}>M_{7,0}$.

\section{Conclusion}
\label{sec6}
Gray codes in $S_n$ under the $\ell_{\infty}$-metric by using the ``push-to-the-top'' operations are very useful in the framework of rank modulation for flash memories. These codes can protect against spike errors in the cells. In this paper, we gave two constructions of $\ell_{\infty}$-snakes which improve on Yehezkeally and Schwartz's construction. On the one hand, we presented one construction of $\ell_{\infty}$-snakes by using cyclic and complete RMGCs. On the other hand, we gave another construction of $\ell_{\infty}$-snakes by using $\mathcal{K}$-snakes. By our constructions, we can obtain longer $\ell_{\infty}$-snakes than Yehezkeally and Schwartz's ones.

\section*{Acknowledgment}
This research is supported by National Key Basic Research Problem of China (Grant No. 2013CB834204) and the National Natural Science Foundation of China (No. 61171082).

\end{document}